%% file: submission_arXiv.tex
\documentclass[conference,a4paper]{IEEEtran}
\IEEEoverridecommandlockouts
\usepackage[a4paper,left=16.5mm,right=16.5mm,top=19mm,bottom=44mm]{geometry}
\usepackage{amsmath,amsfonts}
\usepackage{array}
\usepackage[caption=false]{subfig}
\usepackage{textcomp}
\usepackage{stfloats}
\usepackage{url}
\usepackage{verbatim}
\usepackage{graphicx}
\usepackage{cite}
\usepackage{wrapfig}

\usepackage[utf8]{inputenc}
\usepackage[T1]{fontenc}
\usepackage[bookmarks=false]{hyperref}
\usepackage{amssymb}
\usepackage{amsthm}
\usepackage{enumerate}
\usepackage[inline]{enumitem}
\usepackage{tikz}
\usepackage{booktabs}
\usepackage{tabularx}
\usepackage{pgfplots}
\usetikzlibrary{positioning,shapes.misc,shapes.geometric,patterns,spy,calc}
\usepgfplotslibrary{fillbetween}
\pgfplotsset{compat=newest}
\usepackage[ruled,vlined]{algorithm2e}
\usepackage[capitalize]{cleveref}
\usepackage[normalem]{ulem}
\usepackage{balance}
\usepackage{float}
\usepackage{comment}
\usepackage[nolist,nohyperlinks]{acronym}
\acrodef{ldp}[LDP]{local differential privacy}
\acrodef{dp}[DP]{differential privacy}
\acrodef{ddp}[DDP]{dependent differential privacy}
\acrodef{3r}[3R]{Three-Region}
\acrodef{mq}[MQ]{Markov-Quilt}
\acrodef{mdp}[MDP]{Markov Decision Process}
\acrodef{lrm}[LRM]{Local Redaction Mechanism}
\acrodef{iid}[iid]{independent and identically distributed}

\newtheorem{definition}{Definition}

\newtheorem{lemma}{Lemma}

\newtheorem{corollary}{Corollary}
\newtheorem{theorem}{Theorem}

\input{macros}

\tikzset{
    mdp/.style={color=olive},
    3r-grid/.style={color=bleudefrance,dashdotted},
    3r-exact/.style={color=red,dashed},
    3r-analytic/.style={color=red,densely dotted},
    mq/.style={color=black}
}

\begin{document}

\title{%
    The Right to Be Forgotten under Correlation: Privacy–Utility Tradeoffs
    \thanks{\textbf{Emails:} \{luis.massny,rawad.bitar\}@tum.de, fangweiye@nuaa.edu.cn, salim.elrouayheb@rutgers.edu}
    \thanks{R. Bitar is supported by the German Research Foundation (DFG) grant agreement number BI 2492/1-1. The work of S. El Rouayheb  was supported in part by the Rutgers Global International Collaborative Research Grant.}
}

\author{%
    \IEEEauthorblockN{Luis Massny$^*$, Rawad Bitar$^*$, Fangwei Ye$^\dagger$, Salim El Rouayheb$^\ddagger$}
    \IEEEauthorblockA{
        $*$ \textit{School of Computation, Information and Technology, Technical University of Munich}, Munich, Germany\\
        $\dagger$ \textit{College of Computer Science and Technology, Nanjing University of Aeronautics and Astronautics}, Nanjing, China \\
        $\ddagger$ \textit{ECE Department, Rutgers University}, New Brunswick, NJ, USA
    }
}

\maketitle

\begin{abstract}
  We study privacy-preserving algorithms that redact information from a correlated dataset. The goal is to reduce the information leakage about sensitive parts.
  We analytically characterize the optimal tradeoff between privacy (leakage through correlation) and utility (number of redactions) for local mechanisms and binary Markovian correlations.
  Our result shows how to decompose the mechanism design problem into a sequence of local mechanism design problems and reveals their sequential relation.
  The optimal mechanism is the solution to a Finite-Horizon Deterministic Markov Decision Process.
\end{abstract}


\section{Introduction}
\IEEEPARstart{T}{he} right to be forgotten is a central promise of recent privacy protection legislation such as the European Union's General Data Protection Regulation (GDPR)
and the California Consumer Privacy Act (CCPA).
Data ecosystems~\cite{google2022openDataEcosystem,aimeur2022data} that collect valuable datasets for machine learning training and big data analysis can pose a threat to the privacy, and thus, must implement technologies that comply with these regulations.

When data are correlated, however, fulfilling the promise of being forgotten can be challenging as only redacting the data samples flagged for removal is often not enough. Data correlation emerges in various settings: information or data sequences that evolve over time, such as
activity data~\cite{zufferey2023watchYourWatch},
location traces~\cite{karatzoglou2018markovTrajectories,liuLocationPrivacyIts2018},
financial transactions~\cite{chen2012markovTestFinancial},
or energy demands in smart metering~\cite{torriti2014electricityDemandModel,li2018smartMetering,sankar2013smartMeterPrivacy}, are naturally correlated through their temporal evolution. Furthermore, data from separate individuals connected through social relationships are naturally correlated by these connections~\cite {lucaHowSocialRelationships2010,bonhardKnowingMeKnowing2006}.
Also, data used downstream in a learning algorithm can make intermediate results depend on it, creating correlation and potentially necessitating an unlearning~\cite{bourtoule2021machineUnlearning,yinzhi2015makingSystemsForget} for its information to be removed.
In this work, we advocate for a stronger notion of the right to be forgotten, one that takes a holistic view of privacy by accounting for correlations across the dataset. Under this notion, simply removing the requested data sample is not sufficient; one must also limit the information about that sample that can still be inferred from the remaining correlated data.

\Ac{dp}~\cite{dworkDifferentialPrivacy2006} and \ac{ldp}~\cite{duchiLocalPrivacyStatistical2013}
do not cover correlated datasets~\cite{kiferNoFreeLunch2011,liuDependenceMakesYou2016}.
Therefore, several alternative information leakage measures~\cite{kiferPufferfishFrameworkMathematical2014,yangBayesianDifferentialPrivacy2015,zhaoDependentDifferentialPrivacy2017,ghoshInferentialPrivacyGuarantees2017,wang2025inferentiallyPrivate} and adapted \ac{dp} algorithms~\cite{kiferNoFreeLunch2011,chenCorrelatedNetworkData2014,xiaoProtectingLocationsDifferential2015,zhuCorrelatedDifferentialPrivacy2015,liuDependenceMakesYou2016,songPufferfishPrivacyMechanisms2017a,songCompositionInferentialPrivacy2017,chakrabartiOptimalLocalBayesian2022,bai2026compositionpufferfishprivacy} have been proposed for stronger privacy guarantees.
Our leakage measure is syntactically equivalent to inferential privacy~\cite{ghoshInferentialPrivacyGuarantees2017}, an instance of pufferfish privacy~\cite{kiferPufferfishFrameworkMathematical2014}, which provides explicit information leakage guarantees in correlated settings.

Unlike most of the \ac{dp} literature, which designs additive noise, we ask: \emph{which data samples must be removed from downstream processing to guarantee the right to be forgotten for a desired sample?} Therefore, our mechanisms follow a redact or release approach and are called \emph{redaction mechanisms}. A data sample is either processed noise-free by a downstream procedure (could be a machine learning algorithm or other data analysis procedures) or redacted, thereby, limiting
the trace a redacted sample leaves in the non-redacted data~\cite{garg2020formalizingDataDeletion}.

Prior work~\cite{ye2021onOffCorrelation,yeMechanismsHidingSensitive2022,ye2026perfectPrivacy} has shown that nontrivial redaction mechanisms can achieve perfect privacy if the full dataset is processed sequentially.
We take a different approach and turn towards \acp{lrm}, which process each data sample individually. This enables parallel processing of the data samples, gives small algorithmic complexity, memory efficiency, and is crucial in distributed settings at the expense of relaxing the privacy guarantee. A suboptimal family of such mechanisms named \ac{3r} mechanisms has been studied in~\cite{massny2025closeenough}, which restricts the redactions to a sufficient subset of correlated samples.

\paragraph*{Contributions}
We characterize the optimal privacy-utility tradeoff for \acp{lrm}.
Through a recursive decomposition of a mechanism's leakage, we show how a privacy constraint w.r.t. a sample $X_0$ can be decomposed into a sequence of local leakage constraints w.r.t. correlated samples $X_1$, $X_2$, \dots, $X_n$, respectively.
This allows us to characterize the family of mechanism that an optimal \ac{lrm} lives in, which is characterized by a vector $\leakvec$ of local leakages. We also find that the optimal vector $\leakvec$ is the solution to a Finite-Horizon Deterministic \ac{mdp}, yielding an analytic characterization of the optimal privacy-utility tradeoff.

\section{Problem Setting and Definitions}
\label{sec:setting}

\paragraph*{Notation}
For integers, $m, n$ such that $m<n$, we define $[m,n] \define \{m,m+1,\dots,n\}$ and abbreviate $\range{n} \define [1,n]$. Vectors are represented by bold letters, e.g., $\X$ and $\x$, and sets are represented by calligraphic letters, e.g., $\mathcal{Q}$.
For a vector $X=(X_1,\dots, X_n)$ and a set $\mathcal{Q} \subset \range{n}$, we let $\X_\mathcal{Q}$ be the vector of $X_i$'s indexed by $\mathcal{Q}$, i.e., $\X_\mathcal{Q} \define (X_i)_{i\in\mathcal{Q}}$.
Random variables and random vectors are denoted by upper-case letters $X$ and $\X$, respectively, and their realizations are denoted by lower-case letters, i.e., $x$ and $\x$, respectively.
For a random variable $X \in \mathcal{X}$, we define its support as $\supp(X) \define \{ x \in \mathcal{X} \,|\, \Pr{X=x} > 0 \}$. The support of a random vector is defined accordingly.
For $x \in \{0,1\}$, we denote its complement by $\bar{x} \define 1-x$.
Abusing the notation, we refer to the event $\{X_i=a \,|\, i \in \mathcal{Q}\}$ as $\{ \X_\mathcal{Q}=a \}$.
$\indicator{\cdot}$ denotes the indicator function.

\paragraph*{Problem formulation}Consider samples $X_0,\dots, X_n \in \mathcal{X}$ modeled by discrete random variables that are correlated.
Suppose $X_0$ is flagged for redaction. The goal is to find a probabilistic redaction algorithm $\M \colon \mathcal{X}^n \to \mathcal{Y}^n$, called mechanism, which takes as input $\X \define (X_1,\dots,X_n)$ and returns a partially redacted view $\Y \define (Y_1,\dots,Y_n) \in \mathcal{Y}^n$ in order to limit the information leakage about $X_0$. The mechanism either releases $X_t$ and $Y_t = X_t$ or redacts it and $Y_t = \redact$, i.e., $\mathcal{Y} = \mathcal{X} \cup \{\redact\}$.
We focus on \acp{lrm} that process each $X_t$ individually and for which $\Pr{ \Y=\y | \X=\x } = \prod_{t \in [1,n]} \Pr{ Y_t=y_t | X_t=x_t }$.
The information leakage is quantified as follows.
\begin{definition}[Leakage]
\label{def:privacy}
A mechanism $\M \colon \mathcal{X}^n \to \mathcal{Y}^n$, with input $\X$ and output $\Y$ has a leakage $\leak{X_0 \to \Y}$ about $X_0 \in \mathcal{X}$ defined as
\begin{equation*}
    \leak{X_0 \to \Y} \define \log \sup_{\y \in \supp(\Y), \, (x,x^\prime) \in \mathcal{X}^2} \, \frac{ \Pr{ \Y=\y | X_0=x} }{ \Pr{ \Y=\y | X_0=x^\prime } }.
\end{equation*}
For given $\epsilon$, a mechanism is \emph{private} if
    $\leak{X_0 \to \Y} \leq \epsilon$.
\end{definition}
We aim for maximizing a mechanism's utility defined as follows.
\begin{definition}[Utility]
\label{def:utility}
The utility $\nu$ of a mechanism $\M \colon \mathcal{X}^n \to \mathcal{Y}^n$ with input $\X$ and output $\Y$ is
\begin{equation*}
    \nu \define \frac{1}{n} \Exp{ \sum_{t=1}^{n} \indicator{X_t=Y_t} }.
\end{equation*}
\end{definition}
Small leakage requires redacting data, whereas high utility requires revealing it.
An \ac{lrm} with leakage $\epsilon$ is called optimal if it achieves the maximum utility among all \acp{lrm} with leakage at most $\epsilon$.

We model the correlation between the samples by a Markov chain%
\footnote{We neglect the corner case of independent samples, where privacy holds trivially. For instance, when the correlation is modeled by a time-homogeneous binary Markov chain with state transition probabilities $\alpha$ and $\beta$, the samples are independent when $1-\alpha-\beta=0$.}
$X_0 \text{--} X_1  \text{--} \cdots \text{--} X_n$, which is publicly known.
The setting is depicted in \cref{fig:mc_redaction_setting}. We note that $X_0$ can also be seen as the boundary beyond which a sensitive history $X_0,X_{-1},\dots$ is to be forgotten.%
\footnote{The Markovity together with standard pre-processing inequalities~\cite{dworkAlgorithmicFoundationsDifferential2013,duchiLocalPrivacyStatistical2013} implies that if the leakage about $X_0$ is bounded by $\epsilon$, the leakage about the whole history is bounded by $\epsilon$.}
We focus on the elementary problem for binary alphabets $\mathcal{X}=\{0,1\}$, where our leakage criterion guarantees the indistinguishability between realizations of $X_0$. For larger alphabets, this criterion means to guarantee the indistinguishability between binary features of the data.

\begin{figure}[t]
    \centering
    \resizebox{\linewidth}{!}{\input{mc_redaction_setting_X0}}
    \caption{
        Depiction of the problem setting. $X_0$ is flagged for removal. The goal is to guarantee that any non-redacted samples cause a leakage about $X_0$ of at most $\epsilon$, see \cref{def:privacy}, while retaining as many of the correlated Markovian data samples $X_1,\dots,X_n$ as possible to maintain utility, see \cref{def:utility}. A \acl{lrm} $\M$ operates locally on each sample $X_t$ and probabilistically outputs $Y_t \in \{X_t,\redact\}$ according to a prescribed distribution $\Pr{Y_t|X_t}$.
    }
    \label{fig:mc_redaction_setting}
\end{figure}
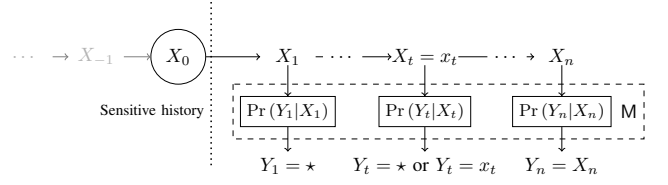

\paragraph*{Definitions}
We view the correlation between the samples through the lens of the so-called pointwise-influence~\cite{massny2025closeenough}.
\begin{definition}
    \label{def:influence}
    The \emph{pointwise-influence} from sample $X_0$ on the realization $X_t=x_t$ is
    \begin{equation*}
    \pwInfl{X_0 \rightsquigarrow X_t=x_t}
    \define
    \log \max_{(x_0,x_0^\prime) \in \mathcal{X}^2} \frac{ \Pr{X_t=x_t | X_0=x_0} }{ \Pr{X_t=x_t | X_0=x_0^\prime} }.
    \end{equation*}
    $\infl{X_0 \rightsquigarrow X_t} \define \max_{x_t \in \mathcal{X}} \pwInfl{X_0 \rightsquigarrow X_t=x_t}$ is called the \emph{max-influence} from $X_0$ on $X_t$.
\end{definition}
The pointwise-influence defines a separation of the samples into three (possibly empty) sets, called \emph{regions}:
\begin{align*}
    \region[L] &\define \{ t \in \range{n} \,|\, \epsilon < \pwInfl{X_0 \rightsquigarrow X_t=x_t} \text{ for all } x_t \in \mathcal{X} \},\\
    \region[S] &\define \{ t \in \range{n} \,|\, \epsilon \geq \pwInfl{X_0 \rightsquigarrow X_t=x_t} \text{ for all } x_t \in \mathcal{X} \},\\
    \region[M] &\define \range{n} \setminus \{\region[L] \cup \region[S]\}.
\end{align*}
To simplify the notation, we denote the first sample that has a pointwise-influence of less than $\epsilon$ as
\begin{equation*}
    \tau \define \min\{t \in \range{n} \,|\, \pwInfl{X_0 \rightsquigarrow X_t=x_t} < \epsilon \text{ for some } x_t \in \mathcal{X} \},
\end{equation*}
and let $\hat{x}_\tau \define \arg\max_{x_\tau} \pwInfl{X_0 \rightsquigarrow X_\tau=x_\tau}$, $\check{x}_\tau = 1-\hat{x}_\tau$.
For $t > \tau$ we then recursively define:
\begin{align}
    \Lambda_t(x_t) \define& \log\tfrac{\Pr{X_t=x_t | X_{t-1}=\hat{x}_{t-1}}}{\Pr{X_t=x_t | X_{t-1}=\check{x}_{t-1}}} \text{ with},
    \\
    \label{eq:check-hat}
    \hat{x}_t \define& \arg\max_{x_t} \Lambda_t(x_t),\, \check{x}_t = 1-\hat{x}_t \\
    \label{eq:threshold}
    \ell_t \define& \textstyle \Lambda_t(\hat{x}_t),
\end{align}
and, with $0 \leq \leakval < \ell_t$,
\begin{align}
    \label{eq:gamma-tau}
  \gamma_\tau(\leakval[\tau]) \define& \textstyle -\log \max_{x \in \mathcal{X}} \frac{
        \Pr{X_\tau = \hat{x}_\tau | X_0=x} - e^{\leakval[\tau]} \Pr{X_\tau = \hat{x}_\tau | X_0=\bar{x}}
    }{
        e^{\leakval[\tau]} \Pr{X_\tau=\check{x}_\tau | X_0=\bar{x}} - \Pr{X_\tau=\check{x}_\tau | X_0=x}
    },
    \\
    \label{eq:gamma}
    \gamma_t(\leakval) \define& \textstyle -\log \frac{ \Pr{X_t=\hat{x}_t|X_{t-1}=\hat{x}_{t-1}} - e^{\leakval} \Pr{X_t=\hat{x}_t|X_{t-1}=\check{x}_{t-1}} }{ e^{\leakval} \Pr{X_t=\check{x}_t|X_{t-1}=\check{x}_{t-1}} - \Pr{X_t=\check{x}_t|X_{t-1}=\hat{x}_{t-1}} }.
\end{align}
The functions $\gamma_t$ are increasing and defined only by the data correlation.
For a more compact notation, we may abbreviate
\begin{equation}
    \label{eq:gamma-ext}
    \tilde{\gamma}_t \colon [0,\infty) \to \mathbb{R} \cup \{\infty\},
    \,\leakval \mapsto
    \begin{cases}
        \gamma_t(\leakval) & \text{if } \leakval < \ell_t \text{ or } t=\tau,
        \\
        \infty & \text{otherwise},
    \end{cases}
\end{equation}
and $\hat{q}_t \define \Pr{Y_t=\redact|X_t=\hat{x}_t}$, $\check{q}_t \define \Pr{Y_t=\redact|X_t=\check{x}_t}$.

\section{Main Results}
\label{sec:main_results}

We find a recursive decomposition (\cref{lma:recurrence-relation}) of the leakage $\epsilon$ for an \ac{lrm} into a sequence of interrelated local leakage constraints $\leakval[\tau]$, $\leakval[\tau+1]$, $\dots$, $\leakval[n]$. We define $\leakval[n+1] \define 0$. This motivates the family of leakage-induced mechanisms.
\begin{definition}
A \emph{leakage-induced mechanism} $\M_{\leakvec}^{(\tau)}$ is defined by a \emph{leakage vector}%
\footnote{For notational convenience, we let $\leakvec \in \mathbb{R}^{n}$, but use only $\leakvec_{[\tau,n]}$.}
$\leakvec \in \mathbb{R}_{\geq 0}^{n}$ with $\leakvec_{[1,\tau-1]}=0$ and uses
$\Pr{\Y_{[1,\tau-1]}=\redact}=1$ and for $t \in [\tau,n]$:
\begin{align*}
    \Pr{Y_t=\redact|X_t=\hat{x}_t} &= \begin{cases}
        1 & \text{ if } \leakval < \ell_t \text{ or } t=\tau,\\
        0 & \text{ otherwise},
    \end{cases} \\
    \Pr{Y_t=\redact|X_t=\check{x}_t} &= \begin{cases}
        e^{-(\tilde{\gamma}_t(\leakval)-\leakval[t+1])} & \text{ if } \leakval < \ell_t,\\
        0 & \text{ otherwise}.
    \end{cases}
\end{align*}
\end{definition}
Designing a leakage-induced mechanism reduces the problem to a core dichotomy: for each sample $X_t$ a mechanism can: either exploit the allowed local leakage $\leakval$ to reduce the redaction probability of $X_t$; or benefit from a lower redaction probability for $X_{t+1},\dots,X_n$.

\begin{theorem}
    \label{thm:leakage-induced-mechanism}
    The leakage-induced mechanism $\M_{\leakvec}^{(\tau)}$ is private if and only if $\leakval[\tau]=\epsilon$ and $0 \leq \leakval \leq \tilde{\gamma}_{t-1}(\leakval[t-1])$ for $\tau+1\leq t \leq n$. The utility-maximizing choice $\leakvec^\star$ gives an optimal \ac{lrm}.
\end{theorem}
We show that the problem of finding a utility-maximizing leakage vector $\leakvec^\star$ is a Finite-Horizon Deterministic \ac{mdp} with state $s_t = \leakval$ and action $a_t = \leakval[t+1]$.
This yields an analytic expression for the privacy-utility tradeoff of an optimal \ac{lrm} $\M_\star$ in terms of the so-called value~\cite{bellman1954theory}, which is recursively defined as
\begin{align}
    \label{eq:optimal-value}
    V_t^{\M_\star}(\leakval) = \max_{\leakval[t+1] \in [0,\tilde{\gamma}_t(\leakval)]} \left[ R_t(\leakval,\leakval[t+1]) + V_{t+1}^{\M_\star}(\leakval[t+1]) \right]
\end{align}
for $\tau \leq t \leq n$, with $V_{n+1}^{\M_\star}(\leakval[n+1])=0$ and
\begin{align}
    \nonumber
    & R_t(\leakval,\leakval[t+1]) \define 1-\hat{q}_t \Pr{X_t=\hat{x}_t}-\check{q}_t \Pr{X_t=\check{x}_t}
    \\\label{eq:reward}
    =&
    \begin{cases}
        1 & \text{ if } \leakval \geq \ell_t, \\
        \Pr{X_t=\check{x}_t} \left[ 1-e^{-(\gamma_t(\leakval)-\leakval[t+1])} \right] & \text{ otherwise}.
    \end{cases}
\end{align}

\begin{corollary}
    \label{eq:optimal-utility}
    The optimal privacy-utility tradeoff is 
    \begin{equation}
    \nu_\star = V_\tau^{\M_\star}(\epsilon) / n.
    \end{equation}
\end{corollary}

We plot the optimal privacy-utility tradeoff in \cref{fig:privacy-utility} for a particular parameter set.

\begin{figure}[t]
    \centering
    \resizebox{0.9\linewidth}{!}{\input{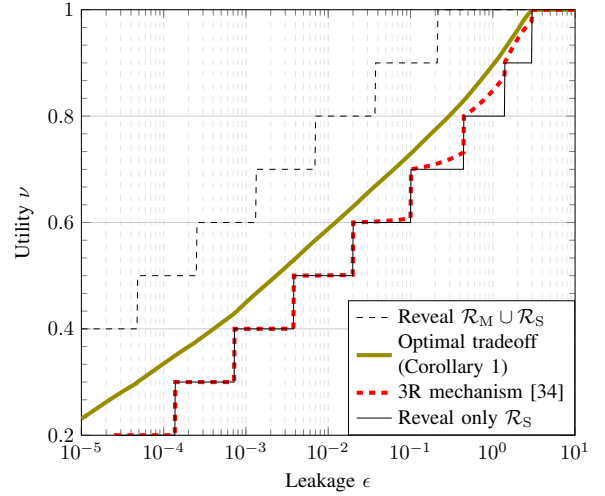}}
    \caption{Optimal privacy-utility tradeoff for \acp{lrm}. For comparison, the utility achievable by always revealing $\region[S]$ (lower bound) or $\region[S] \cup \region[M]$ (upper bound, might exceed leakage $\epsilon$) is given.
    The optimal tradeoff is given by \cref{eq:optimal-utility}.
    We consider $n=10$ samples correlated by a time-homogeneous binary Markov chain with transition probabilities $\Pr{X_{t+1}=1|X_t=0}=10^{-2}$, $\Pr{X_{t+1}=0|X_t=1}=8 \cdot 10^{-1}$ initialized in its stationary distribution.
    }
    \label{fig:privacy-utility}
\end{figure}

\section{Optimal \acl{lrm}}
\label{sec:optimal-mechanisms}
Our main result is built from three layers, which we explain and discuss sequentially.
As the foundation, we study the fundamental properties of the leakage for \acp{lrm} in \cref{sec:general-results} first. We then, build up on this to establish a decomposition law for the leakage in \cref{sec:decomposition} and subsequently outline the proof of \cref{thm:leakage-induced-mechanism}. Finally, this yields a characterization of the optimal privacy-utility tradeoff as discussed in \cref{sec:mdp-interpretation}.

\subsection{General Properties}
\label{sec:general-results}

A mechanism's leakage is discontinuous in the redaction probability: allowing realizations with higher pointwise-influence to be revealed necessarily elevates the leakage. This observation provides insights into the output distribution's shape of an optimal redaction mechanism as follows.
\begin{lemma}
\label{lma:optimal-mechanism-must-restrict-support}
If an \ac{lrm} is optimal, then for every $t \in [1,n]$ the following holds:
\begin{enumerate}[label=(\roman*)]
    \item $\Pr{Y_t=\redact | X_t=x_t} < 1 \,\forall_{x_t \in \mathcal{X}} \!\implies \Pr{Y_t=\redact} = 0$,
    \item $\Pr{Y_t = \redact} = 0 \implies \Pr{Y_s = \redact} = 0 \,\forall_{s \geq t}$.
\end{enumerate}
\begin{proof}
    It can be shown that the leakage of an \ac{lrm} is lower-bounded by the pointwise influence $\pwInfl{X_0 \rightsquigarrow X_t=x_t}$ of any realization $X_t=x_t$ that is revealed with positive probability, cf.~Appendix~\ref{apx:lrm-properties}. Hence, redaction does not affect the leakage unless there is at least one realization that is always redacted, i.e., $\Pr{Y_t=\redact | X_t=x_t} = 1$.
    The second property is an implication of the Markovian data correlation.
    We give the full technical proof of \cref{lma:optimal-mechanism-must-restrict-support} in Appendix~\ref{apx:lrm-properties}.
\end{proof}
\end{lemma}
\begin{corollary}
    \label{cor:optimal-mechanism-must-restrict-support}
    For every $t \in \range{n}$, an optimal \ac{lrm} either never redacts $X_t$ ($\Pr{Y_t=\redact} = 0$)
    or there exists at least one $\hat{x}_t \in \mathcal{X}$ that is always redacted ($\Pr{Y_t=\redact|X_t=\hat{x}_t} = 1$).
    \begin{proof}
        Either $\Pr{Y_t=\redact} = 0$ or $\Pr{Y_t=\redact} > 0$. In the latter case, if $\Pr{Y_t=\redact|X_t=x_t} < 1 \,\forall_{x_t \in \mathcal{X}}$, then \cref{lma:optimal-mechanism-must-restrict-support} implies that $\Pr{Y_t=\redact} = 0$ contradicting the fact that $\Pr{Y_t=\redact} > 0$. Hence, there must exist at least one $\hat{x}_t \in \mathcal{X}$ with $\Pr{Y_t=\redact | X_t=\hat{x}_t} = 1$ in this case.
    \end{proof}
\end{corollary}

Due to \cref{lma:optimal-mechanism-must-restrict-support}, we can restrict to mechanisms that form a contiguous set of $w$ probabilistically redacted samples, i.e. $\{ t \,|\, \Pr{Y_t=\redact}>0 \} = \range{w}$.
We can also assume without loss of optimality
that for all $t \in [1,w]$
\begin{equation}
    \label{eq:wlog_mechanism}
    \Pr{Y_t=\redact|X_t=x_t}
    =
    \begin{cases}
    1 & x_t=\hat{x}_t, \\
    \check{q}_t & x_t=\check{x}_t,
    \end{cases}
\end{equation}
for some $\hat{x}_t$ and $\check{x}_t=1-\hat{x}_t$ (a leakage-induced mechanism chooses according to \cref{eq:check-hat}).
Any output of a redaction mechanism can be written as $\y^{(t)} = (\redact,\dots,\redact,x_{t+1},y_{t+2},\dots,y_n)$ for a $t \in \range{w}$ and $x_{t+1} \in \supp(Y_{t+1}) \cap \mathcal{X}$.
By the Markovity, $\pwInfl{X_0 \rightsquigarrow \Y=\y^{(t)}}$ is independent of $y_{t+2},\dots,y_n$.
Therefore, the pointwise-influence of an output $\y^{(t)}$ is uniquely characterized by its equivalence class
\begin{align*}
    \mathcal{Y}^{(t)}_{x_{t+1}} \define& \{
        \y = (\redact,\dots,\redact,x_{t+1},y_{t+2},\dots,y_n) \\
        &\quad\mid (y_{t+2},\dots,y_n) \in \supp(\Y_{[t+2,n]})
    \}.
\end{align*}
We can write an \ac{lrm}'s leakage as the maximum of leakages over all redaction windows:
\begin{align*}
    \leak{X_0 \to \Y}
    =& \!\!\!\!\! \max_{\substack{t \in [\tau-1,w] \\ x_{t+1} \in \supp(Y_{t+1})}} \, \max_{x \in \mathcal{X}} \, \log
    \frac{ \Pr{ \Y=\y^{(t)} | X_0=x} }{ \Pr{ \Y=\y^{(t)} | X_0=\bar{x} } },
\end{align*}
where $\y^{(t)} \in \mathcal{Y}^{(t)}_{x_{t+1}}$.
Its utility can be expressed as:
\begin{align*}
    \nu
    =& 1-\frac{1}{n} \sum_{t \in \range{w}} (\Pr{X_t=\hat{x}_t}+\check{q}_t\Pr{X_t=\check{x}_t}).
\end{align*}
In summary, we can formulate the general mechanism design problem explicitly as:
\begin{align*}
    P_0:
    \min_{\substack{\tau \leq w \leq n \\ \hat{x}_t, t \in [\tau,w] \\ 0 < \check{q}_t \leq 1, t \in [\tau,w]}}
    & \sum_{t \in [\tau,w]} (\Pr{X_t=\hat{x}_t}+\check{q}_t\Pr{X_t=\check{x}_t})
    \\
    \text{s.t.}
    \quad& \quad \frac{ \Pr{ \Y=\y^{(t)} | X_0=x} }{ \Pr{ \Y=\y^{(t)} | X_0=\bar{x} } } \leq e^\epsilon \\
    &\quad\text{ for all } x \in \mathcal{X},\\
    &\quad\text{ for all } \y^{(t)} \in \mathcal{Y}^{(t)}_{\check{x}_{t+1}},\,t \in [\tau-1,w-1]\\
    &\quad\text{ for all } \y^{(w)} \in \mathcal{Y}^{(w)}_{\check{x}_{w+1}} \cup \mathcal{Y}^{(w)}_{\hat{x}_{w+1}}.
\end{align*}
The optimization problem $P_0$ is nonconvex, the solution of which requires a search among an exponential number of possible combinations, involving discrete optimization variables $w$ and $\hat{x}_t$, and continuous variables $\check{q}_t$, and is computationally and analytically intractable.
The following questions emerge, which we answer in our work:
\begin{enumerate}[label={(Q\arabic*)}]
    \item\label{question:window} \textbf{When should the redaction window $\range{w}$ terminate?}
    \item\label{question:xhat} \textbf{Which realization $\hat{x}_t$ should be redacted?}
    \item\label{question:qt} \textbf{How should one choose $\check{q}_t$ for $t \in \range{w}$?}
\end{enumerate}

The optimal redaction mechanism is trivial for $\infl{X_0 \rightsquigarrow X_\tau} \leq \epsilon$ (i.e., when $\card{\region[M]}=0$), namely, $\range{w}=\region[L]$ and thus, $\check{q}_t=1$ for all $t \in \range{w}$.
We consider only the nontrivial case where $\infl{X_0 \rightsquigarrow X_\tau} = \pwInfl{X_0 \rightsquigarrow X_\tau=\hat{x}_\tau} > \epsilon \geq \pwInfl{X_0 \rightsquigarrow X_\tau=\check{x}_\tau}$ in the sequel. In this case, $\Pr{Y_\tau=\redact|X_\tau=\hat{x}_\tau}=1$ and the constraint in $P_0$ is satisfied for $t=\tau-1$ by definition of $\tau$. Hence, only the constraints for $t \in [\tau,w]$ are nontrivial.

\subsection{Decomposition into Local Leakage Constraints}
\label{sec:decomposition}

\begin{lemma}[Recurrence relation]
    \label{lma:recurrence-relation}
    Suppose $\y^{(t)} \in \supp(\Y)$.
    For all $s$, $\tau+1 \leq s \leq t$, and for $\leakval[s] \geq 0$ it holds:
    \begin{equation}
        \label{eq:recursion-start}
        \frac{ \Pr{\Y_{[s,t+1]}=\y^{(t)}_{[s,t+1]} | X_{s-1}=\hat{x}_{s-1}} }{ \Pr{\Y_{[s,t+1]}=\y^{(t)}_{[s,t+1]} | X_{s-1}=\check{x}_{s-1}} } \leq e^{\leakval[s]}
    \end{equation}
    if and only if $\leakval[s] \geq \ell_s$
    or
    \begin{align}
        \label{eq:recursive-goal}
        & \hat{q}_s \Pr{\Y_{[s+1,t+1]}=\y^{(t)}_{[s+1,t+1]} | X_{s}=\hat{x}_{s}} \\\nonumber
        &- \check{q}_s e^{\gamma_s(\leakval[s])} \Pr{\Y_{[s+1,t+1]}=\y^{(t)}_{[s+1,t+1]} | X_{s}=\check{x}_{s}} \leq 0.
    \end{align}
    \begin{proof}
        The rigorous proof is technical and deferred to Appendix~\ref{proof:recurrence-relation}.
        The main idea is to exploit the Markovity $X_{s-1} \text{ -- } X_s \text{ -- } \Y_{[s,t+1]}$ and write
        \begin{align*}
            & \Pr{\Y_{[s,t+1]}=\y^{(t)}_{[s,t+1]} | X_{s-1}=x_{s-1}} \\
            =& \quad\check{q}_{s} \Pr{\Y_{[s+1,t+1]}=\y^{(t)}_{[s+1,t+1]} | X_{s}=\check{x}_{s}} \\
            &\quad \cdot \Pr{X_{s}=\check{x}_{s}|X_{s-1}=x_{s-1}} \\
            &+ \hat{q}_{s} \Pr{\Y_{[s+1,t+1]}=\y^{(t)}_{[s+1,t+1]} | X_{s}=\hat{x}_{s}} \\
            &\quad \cdot \Pr{X_{s}=\hat{x}_{s}|X_{s-1}=x_{s-1}}.
        \end{align*}
        The remainder of the proof relies on two arguments: 
        \begin{enumerate*}[label=(\roman*)]
            \item Substituting the expression derived above into \cref{eq:recursion-start} shows that the constraint is always satisfied when $\leakval[s]$ is sufficiently large;
            \item If $\leakval[s]$ is small, then $\gamma_s(\leakval[s])$ is well-defined and \cref{eq:recursion-start} is equivalent to \cref{eq:recursive-goal}.
        \end{enumerate*}
    \end{proof}
\end{lemma}

\begin{proof}[Proof of \cref{thm:leakage-induced-mechanism}]
    We give an overview of the proof idea here and defer the technical details to Appendix~\ref{sec:mdp-proof}.
    The proof is inductive and shows for each $\tau+1\leq t \leq n$ that the overall leakage requirement holds if and only if $0 \leq \leakval \leq \tilde{\gamma}_{t-1}(\leakval[t-1])$.
    This proves the privacy of a leakage-induced mechanism according to \cref{thm:leakage-induced-mechanism}.
    Then, we re-formulate the optimization problem $P_0$ into an \ac{mdp}. Hence, the solution of that \ac{mdp} gives an optimal \ac{lrm}. Importantly, the \ac{mdp} can be understood as choosing the utility-maximizing leakage vector $\leakvec^\star$ for a leakage-induced mechanism $\M_{\leakvec}^{(\tau)}$.
    
    We start the induction by showing that, for every output $\y^{(t)} \in \supp(\Y)$, the leakage requirement $\Pr{\Y=\y^{(t)} | X_0=x_0} - e^\epsilon \Pr{\Y=\y^{(t)} | X_0=\bar{x}_0} \leq 0$ is equivalent to
    \begin{align*}
        & \Pr{\Y_{[\tau+1,t+1]}=\y^{(t)}_{[\tau+1,t+1]} | X_\tau=\hat{x}_\tau} \\
        & - \underbrace{ \check{q}_\tau e^{\gamma_\tau(\leakval[\tau])} }_{{\define e^{ \leakval[\tau+1] }}}
        \Pr{\Y_{[\tau+1,t+1]}=\y^{(t)}_{[\tau+1,t+1]} | X_\tau=\check{x}_\tau} \leq 0.
    \end{align*}
    with $\leakval[\tau]=\epsilon$.
    By the change of variable to $\leakval[\tau+1]$, we obtain a condition similar to \cref{eq:recursion-start}.
    From this on, we employ \cref{lma:recurrence-relation} to transform the constraints of $P_0$ into
    \begin{align*}
        &\hat{q}_t \Pr{Y_{t+1}=y^{(t)}_{t+1} | X_{t}=\hat{x}_{t}} \\
        &\quad- \check{q}_t e^{\gamma_t(\leakval)} \Pr{Y_{t+1}=y^{(t)}_{t+1} | X_{t}=\check{x}_{t}} \leq 0.
    \end{align*}
    Similar to above, we can introduce a new variable $\leakval[t+1]$ for each $t$ and define $\check{q}_t e^{\gamma_t(\leakval)} \define e^{\leakval[t+1]}$. At the same time, we note that $\hat{q}_t=1$ is necessary.
    Nonetheless, we note that this holds only if $\leakval < \ell_t$. If $\leakval$ is sufficiently large, then $\check{q}_t=\hat{q}_t=0$ is feasible and optimal instead.
    We re-iterate this line of arguments for $t=\tau+1,\tau+2,\dots,n$.
    
    Note that $\gamma_t(\leakval)$ is well-defined in the case $\leakval < \ell_t$ and defines the feasible region for $\leakval[t+1]$. In the other case, in contrast, the feasible $\leakval[t+1]$ is unconstrained since it is feasible to have $\Pr{Y_{t+1}=\redact}=\cdots=\Pr{Y_{n}=\redact}=0$.
\end{proof}
To arrive at an \ac{mdp} representation, we view the local leakage as the state $s_t=\leakval$ of a process, which evolves as $s_{t+1} = a_t$, when $a_t=\leakval[t+1]$ denotes the action taken in state $s_t$. Namely, the action describes how much leakage is propagated on to the next state.
Re-writing the utility as $\nu = \frac{1}{n}\sum_{t \in [\tau,n]} R_t(s_t,a_t)$, we arrive at an equivalent finite-horizon deterministic \ac{mdp}:
\begin{align}
    \label{eq:mdp}
    \max_{a_t, \tau \leq t \leq n}
        &\sum_{t \in [\tau,n]} R_t(s_t,a_t)
    \\\nonumber
    \text{s.t.}\,
    & s_{t+1} = a_t \text{ and }a_t \in [ 0, \tilde{\gamma}_{t}(s_t) ],\, \tau \leq t \leq n.
\end{align}
The initial and final states are $s_{\tau} = \epsilon$, and $s_{n+1} = 0$.
$R_t(\leakval,\leakval[t+1])$ as defined in \cref{eq:reward} is the reward of transitioning from state $\leakval$ to state $\leakval[t+1]$.
We note that our transformation steps into \eqref{eq:mdp} do not transform the optimization objective. As such, the same steps can be applied to establish a result similar to \cref{thm:leakage-induced-mechanism} under other utility measures. The formulation is a proper \ac{mdp} whenever the utility measure can be written as an additive reward $\sum_{t \in [\tau,n]} R_t(s_t,a_t)$.

\subsection{Interpretation and Characterization}
\label{sec:mdp-interpretation}

The design of a leakage-induced mechanism and its \ac{mdp} representation permit the following alternative views. A private leakage-induced mechanism according to \cref{thm:leakage-induced-mechanism} can be understood as a sequential composition of the leakage. The \ac{mdp} form can be understood as a sequence of decisions: whether to collect an immediate reward by revealing more about the current instance (allow low subsequent leakage $\leakval[t+1]$); or to maintain a high redaction probability in the hope to reveal more about future instances (allow high subsequent leakage $\leakval[t+1]$). An illustration of both perspectives is provided in Appendix~\ref{apx:illustration}. The leakage allowed in subsequent releases can be thought of as a budget that is to be consumed.

As standard for an \ac{mdp}, we can describe the optimization objective in terms of a value function. The value that allowing leakage $\leakval$ provides under a mechanism $\M$ (which takes actions $a_t(\leakval)=\leakval[t+1]$) is defined recursively as:
$V_t^\M(\leakval) = R_t(\leakval,\leakval[t+1]) + V_{t+1}^\M(\leakval[t+1])$.
According to Bellman's principle of optimality~\cite{bellman1954theory}, an optimal mechanism $\M_\star$ satisfies:
\begin{align*}
    V_t^{\M_\star}(\leakval) = \max_{\leakval[t+1] \in [0,\tilde{\gamma}_t(\leakval)]} \left[ R_t(\leakval,\leakval[t+1]) + V_{t+1}^{\M_\star}(\leakval[t+1]) \right]
\end{align*}
for $\tau \leq t \leq n$, with $V_{n+1}^{\M_\star}(\leakval[n+1])=0$.
This yields the optimal privacy-utility tradeoff in~\cref{eq:optimal-utility}.

\section{Discussion and Conclusion}
\label{sec:discussion}

Our result answers the stated questions as follows.

\ref*{question:window} \textbf{When to terminate the redaction window:}
The redaction window is terminated as soon as the leakage budget $\leakval$ saved from previous releases reaches a threshold $\ell_t$.
Whether it pays off to save more leakage budget depends on the function $\gamma_t$ and on the marginal distribution of $X_t$.
It is always possible to find a policy to terminate the redaction window at the end of the region $\region[M]$, but not always optimal.

\ref*{question:xhat} \textbf{Which realization to redact:}
The realization $\hat{x}_\tau$ that is redacted surely corresponds to the high-leaking realization that is influenced most by $X_0$.
In the same spirit, $\hat{x}_t$ for $t > \tau$ corresponds to the high-leaking realization w.r.t. $X_{t-1}=\hat{x}_{t-1}$.
While the high-leaking realization must be redacted necessarily, the low-leaking realization $\check{x}_t$ can be released with positive probability. Nonetheless, the redaction probability must be strictly positive to preserve plausible deniability about $\hat{x}_t$ when $Y_t=\redact$.

\ref*{question:qt} \textbf{How to design the redaction probabilities:}
The optimal redaction probabilities have a threshold-type relation to the allowed leakage $\leakval$.
For leakage $\leakval \geq \ell_t$, the redaction probability vanishes.
For leakage $\leakval < \ell_t$, the redaction probability for sample $X_t$ decays exponentially in the leakage consumed.
The leakage consumption
is upper bounded by $\gamma_t(\leakval)$, which describes the amplification of the leakage budget saved through previous redaction decisions.

We evaluate the optimal privacy-utility tradeoff in \cref{fig:privacy-utility}, by numerically approximating the \ac{mdp} solution through Approximate Dynamic Programming~\cite{powell2011dynamicProgramming}, discretizing the state and action space to $\Delta=100$ equally spaced values.%
\footnote{We implemented the approximate \ac{mdp} solution using the DynaProg MATLAB toolbox~\cite{miretti2021DynaProg}.}
This reveals the sub-optimality of the previously proposed \ac{3r} mechanism~\cite{massny2025closeenough}, for small allowed leakage in particular.

As future work, we pursue the development of efficient algorithms to achieve the optimal privacy-utility tradeoff and the extension to nonbinary alphabets.

\clearpage
\balance
\section*{Statement of AI Use}
The authors acknowledge the use of Claude Opus 5.5 (Anthropic)
for the discovery of related literature and GPT-5.6 (Open AI) for sketching the definition of a leakage-induced mechanism according to our ideas.
The authors take full responsibility for the contents in this manuscript.
\bibliographystyle{ieeetr}
\bibliography{references}

\clearpage
\appendices

\section{General properties of \acp{lrm}}
\label{apx:lrm-properties}

\begin{lemma}
\label{lma:zerohop_redaction}
Let $\mathcal{P} \define \{t \in \range{n} | \Pr{Y_t=\redact}<1\}$ denote the set of samples that are revealed with positive probability. The leakage of an \ac{lrm} about $X_0$ is at most
\begin{equation*}
    \leak{X_0 \to \Y} \leq \infl{X_0 \rightsquigarrow \X_{\mathcal{P}}}.
\end{equation*}
For any subset $\mathcal{Q} \subseteq \mathcal{P}$, if the realization $\x_{\mathcal{Q}} \in \supp(\Y_\mathcal{Q}) \cap \mathcal{X}^{|\mathcal{Q}|}$ is revealed with positive probability, then the leakage of an \ac{lrm} about $X_0$ is at least
\begin{equation*}
    \leak{X_0 \to \Y} \geq \pwInfl{X_0 \rightsquigarrow \X_{\mathcal{Q}}=\x_{\mathcal{Q}}}.
\end{equation*}
\end{lemma}
The above properties follow from the Markovity and the definition of redaction mechanisms.

One can observe that the leakage is confined between the max-influence of released samples and the pointwise-influence of the revealed realizations.
Concretely, \cref{lma:zerohop_redaction} shows that no \ac{lrm} with leakage at most $\epsilon$ can reveal a sample from the region $\region[L]$.
In contrast, revealing all samples from region $\region[S]$ never exceeds the allowed leakage, which serves as a trivial baseline redaction mechanism, cf.~\cite{massny2025closeenough}.
Most interestingly, $\region[M]$ contains samples the disclosure of which, for some realizations, does not exceed the allowed leakage per se, but for other realizations does. The crux of designing a private mechanism boils down to ensuring indistinguishability between realizations of the former kind and the latter kind.

    \begin{proof}[Proof of \cref{lma:zerohop_redaction}]
        Suppose there exists $\y \in \supp(\Y)$ with $y_t \neq \redact$ for all $t \in \mathcal{Q}$. Define $\bar{\mathcal{Q}}=\range{n} \setminus \mathcal{Q}$. It holds:
        \begin{align*}
            & \Pr{\Y=\y | X_0=x}
            \\
            \overset{(a)}{=} & \Pr{\Y=\y | X_0=x, \X_\mathcal{Q}=\y_\mathcal{Q}}
            \Pr{\X_\mathcal{Q}=\y_\mathcal{Q} | X_0=x}
            \\
            \overset{(b)}{=} & \Pr{\Y_{\bar{\mathcal{Q}}}=\y_{\bar{\mathcal{Q}}} | X_0=x, \X_\mathcal{Q}=\y_\mathcal{Q}}
            \\
            & \cdot \Pr{\X_\mathcal{Q}=\y_\mathcal{Q} | X_0=x} \Pr{\Y_\mathcal{Q}=\y_\mathcal{Q} | \X_\mathcal{Q}=\y_\mathcal{Q}}
        \end{align*}
        where $(a)$ follows since $\Pr{Y_t=x_t | X_t=x}=0$ if $x \neq x_t$ for redaction mechanisms, and $(b)$ follows by the Markovity of the chain.
        Hence, we can conclude that
        \begin{align}
            \frac{\Pr{\Y=\y | X_0=x}}{\Pr{\Y=\y | X_0=\bar{x}}}
            =&\, \frac{
                \Pr{\Y_{\bar{\mathcal{Q}}}=\y_{\bar{\mathcal{Q}}} | X_0=x, \X_\mathcal{Q}=\y_\mathcal{Q}}
            }{
                \Pr{\Y_{\bar{\mathcal{Q}}}=\y_{\bar{\mathcal{Q}}} | X_0=\bar{x}, \X_\mathcal{Q}=\y_\mathcal{Q}}
            }
            \label{eq:term1}\\
            &\cdot \frac{ \Pr{\X_\mathcal{Q}=\y_\mathcal{Q} | X_0=x} }{ \Pr{\X_\mathcal{Q}=\y_\mathcal{Q} | X_0=\bar{x}} }.
            \label{eq:term2}
        \end{align}
        For every $x$, we can always choose a $\y$ with $\y_{\bar{\mathcal{Q}}}$ such that in \cref{eq:term1} it holds
        \begin{equation*}
            \frac{
                \Pr{\Y_{\bar{\mathcal{Q}}}=\y_{\bar{\mathcal{Q}}} | X_0=x, \X_\mathcal{Q}=\y_\mathcal{Q}}
            }{
                \Pr{\Y_{\bar{\mathcal{Q}}}=\y_{\bar{\mathcal{Q}}} | X_0=\bar{x}, \X_\mathcal{Q}=\y_\mathcal{Q}}
            }
            \geq 1.
        \end{equation*}
        Otherwise, we have a contradiction:
        \begin{align*}
            & \Pr{\Y_{\bar{\mathcal{Q}}}=\y_{\bar{\mathcal{Q}}} | X_0=x, \X_\mathcal{Q}=\y_\mathcal{Q}}
            \\
            &< \Pr{\Y_{\bar{\mathcal{Q}}}=\y_{\bar{\mathcal{Q}}} | X_0=\bar{x}, \X_\mathcal{Q}=\y_\mathcal{Q}} \forall_{\y_{\bar{\mathcal{Q}}}}
            \\
            \implies& 1 = \sum_{\y_{\bar{\mathcal{Q}}}} \Pr{\Y_{\bar{\mathcal{Q}}}=\y_{\bar{\mathcal{Q}}} | X_0=x, \X_\mathcal{Q}=\y_\mathcal{Q}}
            \\
            &< 
            \sum_{\y_{\bar{\mathcal{Q}}}} \Pr{\Y_{\bar{\mathcal{Q}}}=\y_{\bar{\mathcal{Q}}} | X_0=\bar{x}, \X_\mathcal{Q}=\y_\mathcal{Q}} = 1.
        \end{align*}
        Therefore,
        \begin{align*}
            \leak{X_0 \to \Y} =& \log \max_{\y,x} \frac{\Pr{\Y=\y | X_0=x}}{\Pr{\Y=\y | X_0=\bar{x}}} \\
            \geq& \log \max_{\y_{\mathcal{Q}},x} \frac{ \Pr{\X_\mathcal{Q}=\y_\mathcal{Q} | X_0=x} }{ \Pr{\X_\mathcal{Q}=\y_\mathcal{Q} | X_0=\bar{x}} } \\
            =& \pwInfl{X_0 \rightsquigarrow X_\mathcal{Q}=\y_\mathcal{Q}}.
        \end{align*}
    \end{proof}
    
    \begin{proof}[Proof of \cref{lma:optimal-mechanism-must-restrict-support}]
    Let $t \in [1,n]$ and $x_t \in \mathcal{X}$. Suppose $\Pr{Y_t=\redact | X_t=x_t} < 1$, but also $\Pr{Y_t=\redact | X_t=x_t} > 0$. Then there exists a $\y \in \supp(\Y)$ such that $y_t=x_t$. For this $\y$ it holds that
    \begin{align}
        &\Pr{\Y=\y|X_0=x} \nonumber\\
        =& \sum_{x_t^\prime} \Pr{\Y=\y|X_t=x_t^\prime,X_0=x} \Pr{X_t=x_t^\prime|X_0=x} \nonumber\\
        =& \Pr{\Y=\y|X_t=x_t,X_0=x} \Pr{X_t=x_t|X_0=x} \label{eq:output-distribution-without-redaction}
    \end{align}
    The last equality holds, because a local redaction mechanism has $\Pr{Y_t=x_t | X_t=x_t^\prime} = 0$ when $x_t^\prime \neq x_t$, and thus,
    \begin{align*}
        & \Pr{\Y=\y|X_t=x_t^\prime,X_0=x} \\
        =& \Pr{\Y_{-t}=\y_{-t}|X_t=x_t^\prime,X_0=x} \\
        &\quad\Pr{Y_t=x_t | X_t=x_t^\prime,X_0=x} \\
        =& \Pr{\Y_{-t}=\y_{-t}|X_t=x_t^\prime,X_0=x} \\
        &\quad\Pr{Y_t=x_t | X_t=x_t^\prime} = 0.
    \end{align*}
    As a result, the leakage is
    \begin{align*}
        \leak{X_0 \to \Y}
        =& \log \max_{\substack{x \in \mathcal{X},\\\y \in \supp(\Y)}} \frac{ \Pr{\Y=\y|X_0=x} }{ \Pr{\Y=\y|X_0=\bar{x}} } \\
        \geq& \log \max_{\substack{x \in \mathcal{X}, \\ \y_{-t} \in \supp(\Y_{-t}), \\ y_t \in \mathcal{X}}} \frac{ \Pr{\Y=\y|X_0=x} }{ \Pr{\Y=\y|X_0=\bar{x}} },
    \end{align*}
    where by \cref{eq:output-distribution-without-redaction}, it holds:
    \begin{align*}
        &\frac{ \Pr{\Y=\y|X_0=x} }{ \Pr{\Y=\y|X_0=\bar{x}} } \\
        =& \frac{ \Pr{\Y=\y|X_t=y_t,X_0=x} }{ \Pr{\Y=\y|X_t=y_t,X_0=\bar{x}} }
        \cdot\frac{ \Pr{X_t=x_t|X_0=x} }{ \Pr{X_t=x_t|X_0=\bar{x}} }\\
        =& \frac{ \Pr{\Y_{-t}=\y_{-t}|X_t=y_t,X_0=x} \Pr{Y_t=y_t|X_t=y_t} }{ \Pr{\Y_{-t}=\y_{-t}|X_t=y_t,X_0=\bar{x}} \Pr{Y_t=y_t|X_t=y_t} } \\
        & \cdot\frac{ \Pr{X_t=x_t|X_0=x} }{ \Pr{X_t=x_t|X_0=\bar{x}} } \\
        =& \frac{ \Pr{\Y_{-t}=\y_{-t}|X_t=y_t,X_0=x} }{ \Pr{\Y_{-t}=\y_{-t}|X_t=y_t,X_0=\bar{x}} } \\
        &\cdot \frac{ \Pr{X_t=x_t|X_0=x} }{ \Pr{X_t=x_t|X_0=\bar{x}} }.
    \end{align*}
    Consider another mechanism $\mathcal{M}^\prime$ with outputs $\Y^\prime$ and with the same distributions $\Pr{Y_i^\prime|X_i}=\Pr{Y_i|X_i}$ for $i \neq t$, but with $\Pr{Y_t^\prime=\redact} = 0$. Then, $\supp(\Y_{-t})=\supp(\Y^\prime_{-t})$. The leakage is
    \begin{align*}
        \leak{X_0 \to \Y^\prime}
        &= \log \max_{\substack{x \in \mathcal{X},\\ \y^\prime \in \supp(\Y^\prime)}} \frac{ \Pr{\Y^\prime=\y^\prime|X_0=x} }{ \Pr{\Y^\prime=\y^\prime|X_0=\bar{x}} } \\
        &= \log \max_{\substack{x \in \mathcal{X},\\ \y_{-t}^\prime \in \supp(\Y_{-t}^\prime),\\ y^\prime_t \in \mathcal{X}}} \frac{ \Pr{\Y^\prime=\y^\prime|X_0=x} }{ \Pr{\Y^\prime=\y^\prime|X_0=\bar{x}} },
    \end{align*}
    where, as before, \cref{eq:output-distribution-without-redaction} gives
    \begin{align*}
        & \frac{ \Pr{\Y^\prime=\y^\prime|X_0=x} }{ \Pr{\Y^\prime=\y^\prime|X_0=\bar{x}} } \\
        =& \frac{ \Pr{\Y^\prime=\y^\prime|X_t=y^\prime_t,X_0=x} }{ \Pr{\Y^\prime=\y^\prime|X_t=y^\prime_t,X_0=\bar{x}} }
        \cdot \frac{ \Pr{X_t=x_t|X_0=x} }{ \Pr{X_t=x_t|X_0=\bar{x}} } \\
        =& \frac{ \Pr{\Y^\prime_{-t}=\y^\prime_{-t}|X_t=y^\prime_t,X_0=x} }{ \Pr{\Y^\prime_{-t}=\y^\prime_{-t}|X_t=y^\prime_t,X_0=\bar{x}} } \\
        & \cdot\frac{ \Pr{X_t=x_t|X_0=x} }{ \Pr{X_t=x_t|X_0=\bar{x}} } \\
        =& \frac{ \Pr{\Y^\prime_{-t}=\y^\prime_{-t}|X_t=y^\prime_t,X_0=x} }{ \Pr{\Y^\prime_{-t}=\y^\prime_{-t}|X_t=y^\prime_t,X_0=\bar{x}} } \\
        &\cdot \frac{ \Pr{X_t=x_t|X_0=x} }{ \Pr{X_t=x_t|X_0=\bar{x}} }.
    \end{align*}
    It can be observed that $\leak{X_0 \to \Y^\prime} \leq \leak{X_0 \to \Y}$. Furthermore, $\Pr{Y^\prime_t=\redact} = 0 < \Pr{Y_t=\redact}$ implies that the utility $\nu^\prime$ of $\mathcal{M}^\prime$ is larger than the utility $\nu$ of $\mathcal{M}$. Thus, $\mathcal{M}^\prime$ achieves higher utility than $\mathcal{M}$ with at most the same leakage, which contradicts the fact that $\mathcal{M}$ has an optimal utility $\nu$.

    To prove the last statement, we recall the Markovian relation $X_0 - X_t - \Y_{[t+1,n]}$. Therefore, if a record $X_t$ is always released, i.e., $\Pr{Y_t=\redact}=0$, then $Y_t=X_t$ and, hence,
    \begin{align*}
        \frac{ \Pr{\Y=\y|X_0=x} }{ \Pr{\Y=\y|X_0=\bar{x}} }
        = \frac{ \Pr{\Y_\range{t}=\y_\range{t}|X_0=x} }{ \Pr{\Y_\range{t}=\y_\range{t}|X_0=\bar{x}} }.
    \end{align*}
    Since the leakage does not depend on $\Y_{[t+1,n]}$, the utility is maximized by setting $\Y_{[t+1,n]}=\X_{[t+1,n]}$, i.e., $\Pr{Y_s=\redact}=0$ for all $s > t$.
    \end{proof}

\section{Proof of \cref{lma:recurrence-relation}}
    \label{proof:recurrence-relation}
    Let $t\geq\tau$ and $\tau+1 \leq s \leq t$. Suppose $\y^{(t)} \in \supp(\Y)$. It holds
    \begin{align*}
        & \Pr{\Y_{[s,t+1]}=\y^{(t)}_{[s,t+1]} | X_{s-1}=\hat{x}_{s-1}} \\
            &\quad- e^{\leakval[s]} \Pr{\Y_{[s,t+1]}=\y^{(t)}_{[s,t+1]} | X_{s-1}=\check{x}_{s-1}} \leq 0
        \\
        \Leftrightarrow& \check{q}_{s} \Pr{\Y_{[s+1,t+1]}=\y^{(t)}_{[s+1,t+1]} | X_{s}=\check{x}_{s}} \\
        &\quad \cdot \big[ \Pr{X_{s}=\check{x}_{s}|X_{s-1}=\hat{x}_{s-1}} \\
        &\qquad\qquad- e^{\leakval[s]} \Pr{X_{s}=\check{x}_{s}|X_{s-1}=\check{x}_{s-1}} \big] \\
        +& \hat{q}_{s} \Pr{\Y_{[s+1,t+1]}=\y^{(t)}_{[s+1,t+1]} | X_{s}=\hat{x}_{s}} \\
        &\quad \cdot \big[ \Pr{X_{s}=\hat{x}_{s}|X_{s-1}=\hat{x}_{s-1}} \\
        &\qquad\qquad- e^{\leakval[s]} \Pr{X_{s}=\hat{x}_{s}|X_{s-1}=\check{x}_{s-1}} \big] \leq 0.
    \end{align*}
    
    If $\leakval[s] \geq \max_{x_{s}} \log\frac{\Pr{X_{s}=x_{s} | X_{s-1}=\hat{x}_{s-1}}}{\Pr{X_{s}=x_{s} | X_{s-1}=\check{x}_{s-1}}}$, then both summands are non-positive. Thus, any $\Pr{Y_{s}=\redact|X_{s}=\check{x}_{s}} \geq 0$ and $\Pr{Y_{s}=\redact|X_{s}=\hat{x}_{s}} \geq 0$  are feasible.
    
    Otherwise, we have in this case:
    \begin{align}
        \label{eq:Atau}
        &\Pr{X_{s}=\hat{x}_{s}|X_{s-1}=\hat{x}_{s-1}} \\\nonumber
        -& e^{\leakval[s]} \Pr{X_{s}=\hat{x}_{s}|X_{s-1}=\check{x}_{s-1}} > 0,
    \end{align}
    with $\hat{x}_{s} \define \arg\max_{x_{s}} \log\frac{\Pr{X_{s}=x_{s} | X_{s-1}=\hat{x}_{s-1}}}{\Pr{X_{s}=x_{s} | X_{s-1}=\check{x}_{s-1}}}$.
    This also implies
    \begin{align}
        \label{eq:Btau}
        & \Pr{X_s=\check{x}_s|X_{s-1}=\hat{x}_{s-1}} \\\nonumber
        -& e^{\leakval[s]} \Pr{X_{s}=\check{x}_{s}|X_{s-1}=\check{x}_{s-1}} < 0,
    \end{align}
    since $(\ref{eq:Atau})+(\ref{eq:Btau}) = 1-e^{\leakval[s]} \leq 0$.
    Therefore, we can reformulate the equation above as
    \begin{align}
        \nonumber
        & \check{q}_{s} \Pr{\Y_{[s+1,t+1]}=\y^{(t)}_{[s+1,t+1]} | X_{s}=\check{x}_{s}} \\\nonumber
        &\cdot \textstyle\frac{
            \Pr{X_{s}=\check{x}_{s}|X_{s-1}=\hat{x}_{s-1}} - e^{\leakval[s]} \Pr{X_{s}=\check{x}_{s}|X_{s-1}=\check{x}_{s-1}}
        }{
            \Pr{X_{s}=\hat{x}_{s}|X_{s-1}=\hat{x}_{s-1}} - e^{\leakval[s]} \Pr{X_{s}=\hat{x}_{s}|X_{s-1}=\check{x}_{s-1}}
        }
        \\\nonumber
        &+ \hat{q}_{s} \Pr{\Y_{[s+1,t+1]}=\y^{(t)}_{[s+1,t+1]} | X_{s}=\hat{x}_{s}} \leq 0
        \\\nonumber
        \Leftrightarrow&
        \hat{q}_{s} \Pr{\Y_{[s+1,t+1]}=\y^{(t)}_{[s+1,t+1]} | X_{s}=\hat{x}_{s}}
        \\\nonumber
        &- \check{q}_{s} \Pr{\Y_{[s+1,t+1]}=\y^{(t)}_{[s+1,t+1]} | X_{s}=\check{x}_{s}} \\\nonumber
        &\cdot \textstyle\frac{
            e^{\leakval[s]} \Pr{X_{s}=\check{x}_{s}|X_{s-1}=\check{x}_{s-1}} - \Pr{X_{s}=\check{x}_{s}|X_{s-1}=\hat{x}_{s-1}}
        }{
            \Pr{X_{s}=\hat{x}_{s}|X_{s-1}=\hat{x}_{s-1}} - e^{\leakval[s]} \Pr{X_{s}=\hat{x}_{s}|X_{s-1}=\check{x}_{s-1}}
        }
        \leq 0
        \\
        \nonumber
        \Leftrightarrow&
        \hat{q}_s \Pr{\Y_{[s+1,t+1]}=\y^{(t)}_{[s+1,t+1]} | X_{s}=\hat{x}_{s}} \\\nonumber
        &\quad- \check{q}_s e^{\gamma_s(\leakval[s])} \Pr{\Y_{[s+1,t+1]}=\y^{(t)}_{[s+1,t+1]} | X_{s}=\check{x}_{s}} \leq 0.
    \end{align}
    where
    $$\textstyle\gamma_s(\leakval[s]) \define \log \frac{
        e^{\leakval[s]} \Pr{X_{s}=\check{x}_{s}|X_{s-1}=\check{x}_{s-1}} - \Pr{X_{s}=\check{x}_{s}|X_{s-1}=\hat{x}_{s-1}}
    }{
        \Pr{X_{s}=\hat{x}_{s}|X_{s-1}=\hat{x}_{s-1}} - e^{\leakval[s]} \Pr{X_{s}=\hat{x}_{s}|X_{s-1}=\check{x}_{s-1}}
    }$$
    is well-defined since the argument of the logarithm is strictly positive, cf.~\cref{eq:Atau} and \cref{eq:Btau}.

%%%%%%%%%%%%%%%
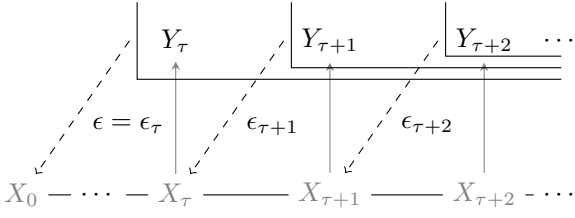
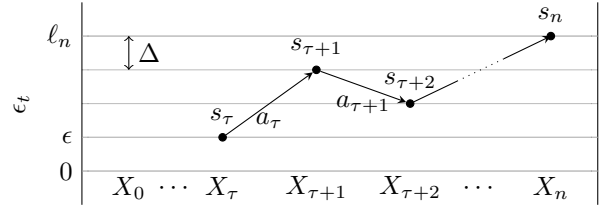
\begin{figure*}[ht]
    \centering
    \subfloat[Sequential composition view: The mechanism design problem is to find a leakage sequence $\{\leakval\}_t$, where each $\leakval$ quantifies the relative leakage of $\Y_{[t,n]}$ w.r.t. the reference $X_{t-1}$, respectively.]{
        \resizebox{0.45\linewidth}{!}{\input{illustration_composition_view}}
        \label{fig:composition-view}
    }
    \hfill
    \subfloat[\ac{mdp} view: The mechanism design problem is to find an optimal sequence of actions $(a_t)_t$ to move between states $(s_t)_t$. Numerical algorithms quantize the states with resolution $\Delta$.]{
        \resizebox{0.45\linewidth}{!}{\input{illustration_mdp_view}}
        \label{fig:mdp-view}
    }
    \caption{Illustration of the perspectives on the mechanism design problem.}
\end{figure*}
%%%%%%%%%%%%%%%

\section{Full proof of \cref{thm:leakage-induced-mechanism}}
\label{sec:mdp-proof}
    We prove the more technical result below.
    \begin{theorem}
    \label{thm:mdp}
    Let $\leakval[\tau]=\epsilon$ and $\leakval[n+1]=0$
    and suppose $\{\leakval\}_{\tau+1 \leq t \leq n}$ are the solution to
    \begin{align*}
        P_1:
        \min_{\{\leakval\}_{\tau+1 \leq t \leq n}}
        \quad& \sum_{t \in [\tau,n]} (\hat{q}_t \Pr{X_t=\hat{x}_t}+\check{q}_t \Pr{X_t=\check{x}_t})
        \\
        \text{s.t.}\quad
        & 0 \leq \leakval \leq \tilde{\gamma}_{t-1}(\leakval[t-1]),\, \tau+1\leq t \leq n,
    \end{align*}
    An optimal binary \ac{lrm} with leakage $\epsilon$ is
    $\Pr{Y_t=\redact}=1$ for all $t \in [1,\tau-1]$ and
    \begin{align*}
        \Pr{Y_t=\redact|X_t=\hat{x}_t} &= \hat{q}_t \define \begin{cases}
            1 & \text{ if } \leakval < \ell_t \text{ or } t=\tau,\\
            0 & \text{ otherwise},
        \end{cases} \\
        \Pr{Y_t=\redact|X_t=\check{x}_t} &= \check{q}_t \define \begin{cases}
            e^{-(\gamma_t(\leakval)-\leakval[t+1])} & \text{ if } \leakval < \ell_t,\\
            0 & \text{ otherwise}.
        \end{cases}
    \end{align*}
    for $t \in [\tau,n]$.
    \end{theorem}
    
    The proof is inductive.
    The starting point is established by similar arguments as used in the proof of \cref{lma:recurrence-relation}.
    A mechanism has leakage $\epsilon$ if and only if
    \begin{equation*}
        \Pr{\Y=\y^{(t)} | X_0=x_0} - e^\epsilon \Pr{\Y=\y^{(t)} | X_0=\bar{x}_0} \leq 0
    \end{equation*}
    holds for all $x_0 \in \mathcal{X}$ and for all $t \geq \tau$ for which $\y^{(t)} \in \supp(\Y)$.
    We first, exploit the Markovity $X_0 \text{ -- } X_\tau \text{ -- } \Y$ to show that
    \begin{align}
    & \Pr{\Y=\y^{(t)} | X_0=x_0} - e^\epsilon \Pr{\Y=\y^{(t)} | X_0=\bar{x}_0} \leq 0
    \nonumber\\\nonumber
    \Leftrightarrow& \Pr{\Y=\y^{(t)}_{[\tau,t+1]} | X_\tau=\hat{x}_\tau} \\
    &\qquad\qquad- e^{\gamma_\tau(\leakval[\tau])} \Pr{\Y=\y^{(t)}_{[\tau,t+1]} | X_\tau=\check{x}_\tau} \leq 0.
    \label{eq:change-of-ref}
    \end{align}
    for $\leakval[\tau] = \epsilon$ and $\gamma_\tau(\leakval[\tau])$ as defined in \cref{eq:gamma-tau}.
    After that, we introduce the variable $\leakval[\tau+1]$ and transform the constraint as follows.
    By the conditional independence of $Y_\tau$ and $\Y_{[\tau+1,n]}$ given $X_\tau=x_\tau$, \cref{eq:change-of-ref} is equivalent to:
    \begin{align}
        & \Pr{\Y=\y^{(t)}_{[\tau+1,t+1]} | X_\tau=\hat{x}_\tau} \\
        \label{eq:epsilon-development}
        & - \underbrace{ \check{q}_\tau e^{\gamma_\tau(\leakval[\tau])} }_{{\define e^{\leakval[\tau+1]}}}
        \Pr{\Y=\y^{(t)}_{[\tau+1,t+1]} | X_\tau=\check{x}_\tau} \leq 0.
    \end{align}
    Here, we use that $\hat{q}_\tau=1$ necessarily.
    Changing the optimization variable from $\check{q}_\tau$ to $\leakval[\tau+1]$, we thus, have $\check{q}_\tau = e^{-(\gamma_\tau(\epsilon)-\leakval[\tau+1])}$. Since $0 < \check{q}_\tau \leq 1$, we require $\leakval[\tau+1] \leq \gamma_\tau(\leakval[\tau])$.
    Since the constraint in \cref{eq:epsilon-development} is infeasible for negative $\leakval[\tau+1]$, we can set $\leakval[\tau+1] \geq 0$ without loss of generality.
    
    The core induction goes as follows. For each $t=\tau+1,\dots,n$, we repeatedly apply \cref{lma:recurrence-relation} to \cref{eq:epsilon-development} (starting from $s=\tau+1$ up to $s=t$), which produces the constraint
    \begin{align}
        \nonumber
        &\hat{q}_t \Pr{Y_{t+1}=y^{(t)}_{t+1} | X_{t}=\hat{x}_{t}} \\
        \label{eq:repetition-end}
        &\quad- \check{q}_t e^{\gamma_t(\leakval)} \Pr{Y_{t+1}=y^{(t)}_{t+1} | X_{t}=\check{x}_{t}} \leq 0.
    \end{align}
    In every iteration $t$, we introduce a variable $\leakval[t+1]$ and determine $\check{q}_t$ and $\hat{q}_t$ as a function of $\leakval$ and $\leakval[t+1]$.
    
    For each $t \geq \tau+1$, if $\leakval \geq \log \frac{\Pr{X_t=\hat{x}_t | X_{t-1}=\hat{x}_{t-1}}}{\Pr{X_t=\hat{x}_t | X_{t-1}=\check{x}_{t-1}}}$, then $\y^{(t-1)}$ with $y^{(t-1)}_{t} = \hat{x}_{t}$ can be in $\supp(\Y)$. Thus, $\check{q}_t=\hat{q}_t=0$ is feasible and is the optimal choice to maximize utility.
    
    Otherwise, if $\leakval < \log\frac{\Pr{X_t=\hat{x}_t | X_{t-1}=\hat{x}_{t-1}}}{\Pr{X_t=\hat{x}_t | X_{t-1}=\check{x}_{t-1}}}$, then $\y^{(t-1)}$ with $y^{(t-1)}_{t} = \hat{x}_t$ cannot be in $\supp(\Y)$, as this would contradict \cref{lma:recurrence-relation}.
    Therefore, it must hold that $\hat{q}_t=1$ and $\check{q}_t$ is defined through $\gamma_t(\leakval)$ accordingly.
    Note that, for the same reason, $\check{q}_t > 0$ is necessary, as otherwise, the symbol $\redact$ would only be a renaming of $\hat{x}_t$ (thus, violating the leakage requirements again). Hence, we can define $\check{q}_t e^{\gamma_t(\leakval)} \define e^{\leakval[t+1]}$. Rewriting this, gives $\check{q}_t = e^{-(\gamma_t(\leakval)-\leakval[t+1])}$. As $\check{q}_t \leq 1$, we require $\leakval[t+1] \leq \gamma_t(\leakval)$.
    Since \cref{eq:repetition-end} is infeasible for $\leakval[t+1] < 0$, we use $\leakval[t+1] \geq 0$ without loss of optimality.
    
    %%%%%%%%%%%%%%
    Note that $\gamma_t(\leakval)$ is well-defined in the case $\leakval < \log\frac{\Pr{X_t=\hat{x}_t | X_{t-1}=\hat{x}_{t-1}}}{\Pr{X_t=\hat{x}_t | X_{t-1}=\check{x}_{t-1}}}$ and defines the feasible region for $\leakval[t+1]$. In the other case, in contrast, the feasible $\leakval[t+1]$ is unconstrained since the leakage requirement becomes independent from $\Y_{[t+1,n]}$, thus allowing for $\Pr{Y_{t+1}=\redact}=\cdots=\Pr{Y_{n}=\redact}=0$.

\section{Composition vs. MDP perspective}
\label{apx:illustration}

The design of a leakage-induced mechanism and its \ac{mdp} representation permit the following alternative views. A private leakage-induced mechanism according to \cref{thm:leakage-induced-mechanism} can be understood as a sequential composition of the leakage. The \ac{mdp} form can be understood as a sequence of decisions: whether to collect an immediate reward by revealing more about the current instance (allow low subsequent leakage $\leakval[t+1]$); or to maintain a high redaction probability in the hope to reveal more about future instances (allow high subsequent leakage $\leakval[t+1]$). The two perspectives are illustrated in \cref{fig:composition-view} and \cref{fig:mdp-view}, respectively.

\end{document}

%% file: macros.tex
\definecolor{asparagus}{rgb}{0.53, 0.66, 0.42}
\definecolor{bostonuniversityred}{rgb}{0.8, 0.0, 0.0}
\definecolor{bleudefrance}{rgb}{0.19, 0.55, 0.91}

\DeclareSymbolFont{bbold}{U}{bbold}{m}{n}
\DeclareSymbolFontAlphabet{\mathbbold}{bbold}

\DeclareMathOperator{\supp}{supp}

\newcommand{\M}{\mathsf{M}}
\newcommand{\X}{\mathbf{X}}
\newcommand{\x}{\mathbf{x}}
\newcommand{\Y}{\mathbf{Y}}
\newcommand{\y}{\mathbf{y}}

\renewcommand{\Pr}[1]{{\operatorname{Pr}\left( {#1} \right)}}
\newcommand{\range}[1]{{\left[ {#1} \right]}}
\newcommand{\card}[1]{{\left| {#1} \right|}}
\newcommand{\indicator}[1]{{\mathbbold{1}\!\left\{ {#1} \right\}}}

\newcommand{\Exp}[2][]{{\ensuremath{\mathbb{E}_{#1}\left[{#2}\right]}}}

\newcommand{\leak}[2][]{{\mathcal{L}_{#1}(#2)}}
\newcommand{\infl}[1]{{\mathcal{I}(#1)}}
\newcommand{\pwInfl}[1]{{i(#1)}}

\newcommand{\leakvec}{\boldsymbol{\varepsilon}}
\newcommand{\leakval}[1][t]{{\varepsilon_{{#1}}}}

\newcommand{\region}[1][]{\ensuremath{{\mathcal{R}_{\mathrm{#1}}}}}

\newcommand{\define}{\ensuremath{\triangleq}}

\newcommand{\redact}{\star}

%% file: mc_redaction_setting_X0.tex
\begin{tikzpicture}
    \tikzset{
        sample/.style={minimum width=1cm,inner sep=0}
    }
    \node[draw,circle,fill=white,sample] (xp) at (-2,0) {$X_0$};
    \node[sample] (x1) at (0,0) {$X_1$};
    \node (dots1) at (1,0) {$\dots$};
    \node[sample] (xt) at (2.5,0) {$X_t=x_t$};
    \node (dots2) at (4,0) {$\dots$};
    \node[sample] (xn) at (5,0) {$X_n$};
    \node[sample,black!30] (x-1) at (-3.5,0) {$X_{-1}$};
    \node[sample,black!30] (x-t) at (-4.8,0) {$\dots$};

    \node[draw,rectangle, font=\small] (m1) at (0,-1) {$\Pr{Y_1 | X_1}$};
    \node[draw,rectangle, font=\small] (mt) at (2.5,-1) {$\Pr{Y_t | X_t}$};
    \node[draw,rectangle, font=\small] (mn) at (5,-1) {$\Pr{Y_n | X_n}$};
    
    \node (y1) at (0,-2) {$Y_1 = \redact$};
    \node (yt) at (2.5,-2) {$Y_t = \redact$ or $Y_t=x_t$};
    \node (yn) at (5,-2) {$Y_n = X_n$};

    \draw (x-t) edge[->,black!50] (x-1);
    \draw (x-1) edge[->,black!50] (xp);
    \draw (xp) edge[->] (x1);
    \draw (x1) edge[-] (dots1);
    \draw (dots1) edge[->] (xt);
    \draw (xt) edge[-] (dots2);
    \draw (dots2) edge [->] (xn);
    
    \draw[->] (x1) edge (m1) (m1) edge (y1);
    \draw[->] (xt) edge (mt) (mt) edge (yt);
    \draw[->] (xn) edge (mn) (mn) edge (yn);

    \draw[dashed] (-1,-0.5) rectangle (6.5,-1.5);
    \node[anchor=east] (m) at (6.5,-1) {$\M$};

    \draw[dotted,thick] (-1.4,1) -- (-1.4,-2) ;
    \node[draw=none,anchor=east] at (-1.4,-1) {\footnotesize Sensitive history};
\end{tikzpicture}

%% file: illustration_composition_view.tex
\begin{tikzpicture}
    \tikzset{
        record/.style={draw=none,opacity=0.5},
        output/.style={draw=none},
        leakage/.style={->,dashed},
    }
    \node[record] (x0) at (0,0) {$X_0$};
    \node[draw=none] (dots) at (1,0) {$\cdots$};
    \node[record] (x1) at (2,0) {$X_{\tau}$};
    \node[output] (y1) at (2,2) {$Y_{\tau}$};
    \node[record] (x2) at (4,0) {$X_{\tau+1}$};
    \node[output] (y2) at (4,2) {$Y_{\tau+1}$};
    \node[record] (x3) at (6,0) {$X_{\tau+2}$};
    \node[output] (y3) at (6,2) {$Y_{\tau+2}$};
    \node[draw=none] (xend) at (7,0) {$\cdots$};
    \node[draw=none] (yend) at (7,2) {$\cdots$};

    \draw % correlations
        (x0) edge (dots)
        (dots) edge (x1)
        (x1) edge (x2)
        (x2) edge (x3)
        (x3) edge (xend);
    \draw[-stealth,opacity=0.5] % outputs
        (x1) edge (y1)
        (x2) edge (y2)
        (x3) edge (y3);
    \draw % \epsilon_\tau
        (1.5,2.5) -- (1.5,1.5) -- (7,1.5)
        (1.4,2) edge[leakage] node[pos=0.5,anchor=north west]{$\epsilon=\epsilon_\tau$} (x0);
    \draw % \epsilon_{\tau+1}
        (3.5,2.5) -- (3.5,1.65) -- (7,1.65)
        (3.4,2) edge[leakage] node[pos=0.5,anchor=north west]{$\epsilon_{\tau+1}$} (x1);
        \draw % \epsilon_{\tau+2}
        (5.5,2.5) -- (5.5,1.8) -- (7,1.8)
        (5.4,2) edge[leakage] node[pos=0.5,anchor=north west]{$\epsilon_{\tau+2}$} (x2);
\end{tikzpicture}

%% file: illustration_mdp_view.tex
\begin{tikzpicture}
    \tikzset{
        record/.style={draw=none},
        state/.style={draw=black,fill,circle, inner sep=1},
    }
    \begin{axis}[
        xmin=-1,
        xmax=10,
        ymin=-2,
        ymax=10,
        unit vector ratio=1 0.6,
        xlabel={$t$},
        ylabel={$\epsilon_t$},
        hide x axis,
        grid=major,
        %ytick distance=1,
        ytick={0,2,4,6,8},
        ytick style={dotted},
        yticklabels=\empty,
        extra y ticks={0,2,8},
        extra y tick labels={$0$,$\epsilon$,$\ell_n$},
        x grid style={line width=.2pt, draw=black!10},
    ]
        \node[record] (x0) at (axis cs:0,-1) {$X_0$};
        \node[record] (dots1) at (axis cs:1,-1) {$\cdots$};
        \node[record] (x1) at (axis cs:2,-1) {$X_\tau$};
        \node[record] (x2) at (axis cs:4,-1) {$X_{\tau+1}$};
        \node[record] (x3) at (axis cs:6,-1) {$X_{\tau+2}$};
        \node[record] (dots2) at (axis cs:7.5,-1) {$\cdots$};
        \node[record] (xn) at (axis cs:9,-1) {$X_n$};

        \node[state,label=above:{$s_\tau$}] (eps1) at (axis cs:2,2) {};
        \node[state,label=above:{$s_{\tau+1}$}] (eps2) at (axis cs:4,6) {};
        \node[state,label=above:{$s_{\tau+2}$}] (eps3) at (axis cs:6,4) {};
        \node[state,label=above:{$s_n$}] (epsn) at (axis cs:9,8) {};

        \draw[-stealth]
            (eps1) edge node[below] {$a_\tau$} (eps2)
            (eps2) edge node[below] {$a_{\tau+1}$} (eps3);
        \draw
            (eps3) edge[-] (axis cs:7,5.33)
            (axis cs:7,5.33) edge[-,dotted] (axis cs:8,6.66)
            (axis cs:8,6.66) edge[-stealth] (axis cs:9,8);
        \draw (axis cs:0,8) edge[<->] node[right] {$\Delta$} (axis cs:0,6);
    \end{axis}
\end{tikzpicture}

%% file: references.bib
@misc{google2022openDataEcosystem,
title={Building the most open data cloud ecosystem},
author={Kazmaier, Gerrit},
year={2022},
howpublished={\url{cloud.google.com/blog/products/data-analytics/building-most-open-data-cloud-all-data-all-source-any-platform}},
url={cloud.google.com/blog/products/data-analytics/building-most-open-data-cloud-all-data-all-source-any-platform},
note = {Accessed: 2026/07/23},
urldate={2026-07-23}
}

@article{aimeur2022data,
title={How data brokers endanger privacy},
author={Aimeur, Esma and Brassard, Gilles and Guo, Muxue},
journal={Transactions on Data Privacy},
year={2022}
}

@inproceedings{zufferey2023watchYourWatch,
author = {No{\'e} Zufferey and Mathias Humbert and Romain Tavenard and Kevin Huguenin},
title = {Watch your Watch: Inferring Personality Traits from Wearable Activity Trackers},
booktitle = {32nd USENIX Security Symposium},
year = {2023},
isbn = {978-1-939133-37-3},
}

@article{karatzoglou2018markovTrajectories,
author = {Karatzoglou, Antonios and Köhler, Dominik and Beigl, Michael},
title = {Semantic-Enhanced Multi-Dimensional Markov Chains on Semantic Trajectories for Predicting Future Locations},
journal = {Sensors},
volume = {18},
year = {2018},
number = {10},
article-number = {3582},
url = {https://www.mdpi.com/1424-8220/18/10/3582},
issn = {1424-8220},
doi = {10.3390/s18103582}
}

@article{liuLocationPrivacyIts2018,
title = {Location {{Privacy}} and {{Its Applications}}: {{A Systematic Study}}},
shorttitle = {Location {{Privacy}} and {{Its Applications}}},
author = {Liu, Bo and Zhou, Wanlei and Zhu, Tianqing and Gao, Longxiang and Xiang, Yong},
year = {2018},
journal = {IEEE Access},
volume = {6},
issn = {2169-3536},
doi = {10.1109/ACCESS.2018.2822260},
}

@article{chen2012markovTestFinancial,
title={TESTING FOR THE MARKOV PROPERTY IN TIMESERIES},
volume={28},
doi={10.1017/S0266466611000065},
number={1},
journal={Econometric Theory},
author={Chen, Bin and Hong, Yongmiao},
year={2012},
}

@article{torriti2014electricityDemandModel,
title = {A review of time use models of residential electricity demand},
journal = {Renewable and Sustainable Energy Reviews},
volume = {37},
year = {2014},
issn = {1364-0321},
doi = {https://doi.org/10.1016/j.rser.2014.05.034},
author = {Jacopo Torriti},
}

@article{li2018smartMetering,
author={Li, Simon and Khisti, Ashish and Mahajan, Aditya},
journal={IEEE Transactions on Information Theory}, 
title={Information-Theoretic Privacy for Smart Metering Systems with a Rechargeable Battery}, 
year={2018},
volume={64},
number={5},
doi={10.1109/TIT.2018.2809005}
}

@article{sankar2013smartMeterPrivacy,
author={Sankar, Lalitha and Rajagopalan, S. Raj and Mohajer, Soheil and Poor, H.V.},
journal={IEEE Transactions on Smart Grid}, 
title={Smart Meter Privacy: A Theoretical Framework}, 
year={2013},
volume={4},
number={2},
doi={10.1109/TSG.2012.2211046}
}

@inproceedings{lucaHowSocialRelationships2010,
title = {How Social Relationships Affect User Similarities},
booktitle = {Proceedings of the {{Social Recommender Systems Workshop}}, in Conjunction with the 2010 {{International Conference}} on {{Intelligent User Interfaces}}},
author = {Luca, Ernesto William De and Said, Alan and Albayrak, Sahin},
year = {2010},
}

@article{bonhardKnowingMeKnowing2006,
title = {'{{Knowing}} Me, Knowing You' --- {{Using}} Profiles and Social Networking to Improve Recommender Systems},
author = {Bonhard, P. and Sasse, M. A.},
year = {2006},
journal = {BT Technology Journal},
volume = {24},
number = {3},
issn = {1358-3948, 1573-1995},
doi = {10.1007/s10550-006-0080-3},
}

@inproceedings{bourtoule2021machineUnlearning,
author={Bourtoule, Lucas and Chandrasekaran, Varun and Choquette-Choo, Christopher A. and Jia, Hengrui and Travers, Adelin and Zhang, Baiwu and Lie, David and Papernot, Nicolas},
booktitle={2021 IEEE Symposium on Security and Privacy},
title={Machine Unlearning}, 
year={2021},
volume={},
number={},
doi={10.1109/SP40001.2021.00019}
}

@inproceedings{yinzhi2015makingSystemsForget,
author={Cao, Yinzhi and Yang, Junfeng},
booktitle={2015 IEEE Symposium on Security and Privacy}, 
title={Towards Making Systems Forget with Machine Unlearning}, 
year={2015},
volume={},
number={},
doi={10.1109/SP.2015.35}
}

@inproceedings{dworkDifferentialPrivacy2006,
title = {Differential {{Privacy}}},
booktitle = {Automata, {{Languages}} and {{Programming}}},
author = {Dwork, Cynthia},
editor = {Bugliesi, Michele and Preneel, Bart and Sassone, Vladimiro and Wegener, Ingo},
year = {2006},
series = {Lecture {{Notes}} in {{Computer Science}}},
publisher = {Springer},
address = {Berlin, Heidelberg},
doi = {10.1007/11787006_1},
isbn = {978-3-540-35908-1},
}

@inproceedings{duchiLocalPrivacyStatistical2013,
title = {Local {{Privacy}} and {{Statistical Minimax Rates}}},
booktitle = {2013 {{IEEE}} 54th {{Annual Symposium}} on {{Foundations}} of {{Computer Science}}},
author = {Duchi, John C. and Jordan, Michael I. and Wainwright, Martin J.},
year = {2013},
issn = {0272-5428},
doi = {10.1109/FOCS.2013.53},
}

@inproceedings{kiferNoFreeLunch2011,
title = {No Free Lunch in Data Privacy},
booktitle = {Proceedings of the 2011 {{ACM SIGMOD International Conference}} on {{Management}} of Data},
author = {Kifer, Daniel and Machanavajjhala, Ashwin},
year = {2011},
publisher = {Association for Computing Machinery},
doi = {10.1145/1989323.1989345},
isbn = {978-1-4503-0661-4},
}

@inproceedings{liuDependenceMakesYou2016,
title = {Dependence {{Makes You Vulnerable}}: {{Differential Privacy Under Dependent Tuples}}},
shorttitle = {Dependence {{Makes You Vulnerable}}},
booktitle = {Proceedings 2016 {{Network}} and {{Distributed System Security Symposium}}},
author = {Liu, Changchang and Chakraborty, Supriyo and Mittal, Prateek},
year = {2016},
publisher = {Internet Society},
doi = {10.14722/ndss.2016.23279},
isbn = {978-1-891562-41-9},
}

@article{kiferPufferfishFrameworkMathematical2014,
title = {Pufferfish: {{A}} Framework for Mathematical Privacy Definitions},
shorttitle = {Pufferfish},
author = {Kifer, Daniel and Machanavajjhala, Ashwin},
year = {2014},
journal = {ACM Transactions on Database Systems},
volume = {39},
number = {1},
issn = {0362-5915},
doi = {10.1145/2514689},
}

@inproceedings{yangBayesianDifferentialPrivacy2015,
title = {Bayesian {{Differential Privacy}} on {{Correlated Data}}},
booktitle = {Proceedings of the 2015 {{ACM SIGMOD International Conference}} on {{Management}} of {{Data}}},
author = {Yang, Bin and Sato, Issei and Nakagawa, Hiroshi},
year = {2015},
publisher = {Association for Computing Machinery},
doi = {10.1145/2723372.2747643},
isbn = {978-1-4503-2758-9},
}

@inproceedings{zhaoDependentDifferentialPrivacy2017,
title = {Dependent {{Differential Privacy}} for {{Correlated Data}}},
booktitle = {2017 {{IEEE Globecom Workshops}}},
author = {Zhao, Jun and Zhang, Junshan and Poor, H. Vincent},
year = {2017},
doi = {10.1109/GLOCOMW.2017.8269219},
}

@inproceedings{ghoshInferentialPrivacyGuarantees2017,
title = {Inferential Privacy Guarantees for Differentially Private Mechanisms},
booktitle = {8th Innovations in Theoretical Computer Science Conference},
author = {Ghosh, Arpita and Kleinberg, Robert},
editor = {Papadimitriou, Christos H.},
year = {2017},
series = {Leibniz International Proceedings in Informatics (Lipics)},
volume = {67},
publisher = {Schloss Dagstuhl--Leibniz-Zentrum fuer Informatik},
issn = {1868-8969},
doi = {10.4230/LIPIcs.ITCS.2017.9},
isbn = {978-3-95977-029-3},
}

@inproceedings{wang2025inferentiallyPrivate,
author = {Wang, Shuaiqi and Zheng, Shuran and Lin, Zinan and Fanti, Giulia and Wu, Zhiwei Steven},
title = {Inferentially-Private Private Information},
year = {2025},
isbn = {9798400712746},
publisher = {Association for Computing Machinery},
doi = {10.1145/3696410.3714702},
booktitle = {Proceedings of the ACM Web Conference 2025},
}

@article{chenCorrelatedNetworkData2014,
title = {Correlated Network Data Publication via Differential Privacy},
author = {Chen, Rui and Fung, Benjamin C. M. and Yu, Philip S. and Desai, Bipin C.},
year = {2014},
journal = {The VLDB Journal},
volume = {23},
number = {4},
issn = {0949-877X},
doi = {10.1007/s00778-013-0344-8},
}

@inproceedings{xiaoProtectingLocationsDifferential2015,
title = {Protecting {{Locations}} with {{Differential Privacy}} under {{Temporal Correlations}}},
booktitle = {Proceedings of the 22nd {{ACM SIGSAC Conference}} on {{Computer}} and {{Communications Security}}},
author = {Xiao, Yonghui and Xiong, Li},
year = {2015},
publisher = {Association for Computing Machinery},
doi = {10.1145/2810103.2813640},
isbn = {978-1-4503-3832-5},
}

@article{zhuCorrelatedDifferentialPrivacy2015,
title = {Correlated {{Differential Privacy}}: {{Hiding Information}} in {{Non-IID Data Set}}},
shorttitle = {Correlated {{Differential Privacy}}},
author = {Zhu, Tianqing and Xiong, Ping and Li, Gang and Zhou, Wanlei},
year = {2015},
journal = {IEEE Transactions on Information Forensics and Security},
volume = {10},
number = {2},
issn = {1556-6021},
doi = {10.1109/TIFS.2014.2368363},
}

@inproceedings{songPufferfishPrivacyMechanisms2017a,
title = {Pufferfish {{Privacy Mechanisms}} for {{Correlated Data}}},
booktitle = {Proceedings of the 2017 {{ACM International Conference}} on {{Management}} of {{Data}}},
author = {Song, Shuang and Wang, Yizhen and Chaudhuri, Kamalika},
year = {2017},
publisher = {ACM},
doi = {10.1145/3035918.3064025},
isbn = {978-1-4503-4197-4},
}

@inproceedings{songCompositionInferentialPrivacy2017,
author={Song, Shuang and Chaudhuri, Kamalika},
booktitle={2017 55th Annual Allerton Conference on Communication, Control, and Computing}, 
title={Composition properties of inferential privacy for time-series data}, 
year={2017},
volume={},
number={},
doi={10.1109/ALLERTON.2017.8262823}}

@inproceedings{chakrabartiOptimalLocalBayesian2022,
author = {Chakrabarti, Darshan and Gao, Jie and Saraf, Aditya and Schoenebeck, Grant and Yu, Fang-Yi},
title = {Optimal Local Bayesian Differential Privacy over Markov Chains},
year = {2022},
isbn = {9781450392136},
publisher = {International Foundation for Autonomous Agents and Multiagent Systems},
booktitle = {Proceedings of the 21st International Conference on Autonomous Agents and Multiagent Systems},
numpages = {3},
location = {Virtual Event, New Zealand},
}

@inproceedings{bai2026compositionpufferfishprivacy,
author = {Bai, Jiamu and He, Guanlin and Gu, Xin and Kifer, Daniel and Maeng, Kiwan},
title = {Composition for Pufferfish Privacy},
year = {2026},
issue_date = {June 2026},
publisher = {VLDB Endowment},
volume = {19},
no-number = {10},
issn = {2150-8097},
doi = {10.14778/3828612.3828632},
booktitle = {Proceedings of the VLDB Endowment},
numpages = {13}
}

@inproceedings{garg2020formalizingDataDeletion,
author="Garg, Sanjam
and Goldwasser, Shafi
and Vasudevan, Prashant Nalini",
editor="Canteaut, Anne
and Ishai, Yuval",
title="Formalizing Data Deletion in the Context of the Right to Be Forgotten",
booktitle="Advances in Cryptology -- EUROCRYPT 2020",
year="2020",
publisher="Springer International Publishing",
address="Cham",
}

@article{ye2021onOffCorrelation,
author={Ye, Fangwei and Naim, Carolina and Rouayheb, Salim El},
journal={IEEE Transactions on Information Theory}, 
title={ON-OFF Privacy in the Presence of Correlation}, 
year={2021},
volume={67},
number={11},
doi={10.1109/TIT.2021.3107944}
}

@article{yeMechanismsHidingSensitive2022,
title = {Mechanisms for {{Hiding Sensitive Genotypes With Information-Theoretic Privacy}}},
author = {Ye, Fangwei and Cho, Hyunghoon and El Rouayheb, Salim},
year = {2022},
journal = {IEEE Transactions on Information Theory},
volume = {68},
number = {6},
issn = {1557-9654},
doi = {10.1109/TIT.2022.3156276},
}

@inproceedings{ye2026perfectPrivacy,
author={Ye, Fangwei and Liu, Zonghong and Parag, Parimal and El Rouayheb, Salim},
booktitle={2026 IEEE International Symposium on Information Theory}, 
title={Perfect Privacy and Strong Stationary Times for Markovian Sources}, 
year={2026},
volume={},
number={},
doi={10.1109/ISIT62367.2026.11653917}
}

@inproceedings{massny2025closeenough,
author={Maßny, Luis and Bitar, Rawad and Ye, Fangwei and El Rouayheb, Salim},
booktitle={2025 IEEE Information Theory Workshop}, 
title={Between Close Enough to Reveal and Far Enough to Protect: a New Privacy Region for Correlated Data}, 
year={2025},
volume={},
number={},
doi={10.1109/ITW62417.2025.11240504}
}

@article{dworkAlgorithmicFoundationsDifferential2013,
title = {The {{Algorithmic Foundations}} of {{Differential Privacy}}},
author = {Dwork, Cynthia and Roth, Aaron},
year = {2013},
journal = {Foundations and Trends{\textregistered} in Theoretical Computer Science},
volume = {9},
number = {3-4},
issn = {1551-305X, 1551-3068},
doi = {10.1561/0400000042},
}

@article{bellman1954theory,
author    = {Bellman, Richard},
title     = {The Theory of Dynamic Programming},
journal   = {Bulletin of the American Mathematical Society},
volume    = {60},
number    = {6},
pages     = {503--515},
year      = {1954},
publisher = {American Mathematical Society},
doi       = {10.1090/S0002-9904-1954-09848-8}
}

@book{powell2011dynamicProgramming,
author={Powell, W. B.},
title={Approximate Dynamic Programming: Solving the Curses of Dimensionality},
year={2011},
publisher={John Wiley \& Sons, Incorporated},
address={Newark}
}

@article{miretti2021DynaProg,
title = {DynaProg: Deterministic Dynamic Programming solver for finite horizon multi-stage decision problems},
journal = {SoftwareX},
volume = {14},
year = {2021},
issn = {2352-7110},
doi = {https://doi.org/10.1016/j.softx.2021.100690},
url = {https://www.sciencedirect.com/science/article/pii/S2352711021000352},
author = {Federico Miretti and Daniela Misul and Ezio Spessa},
}
